\RequirePackage{amsmath}
\documentclass[runningheads]{llncs}
\usepackage{lmodern}
\usepackage[utf8]{inputenc}
\usepackage{latexsym}
\usepackage[T1]{fontenc}
\usepackage{amssymb}
\usepackage{stmaryrd}
\usepackage{eulervm}
\usepackage{palatino}
\usepackage{mdwmath}
\usepackage{mdwtab}

\usepackage{graphicx}
\usepackage[cmyk]{xcolor}
\usepackage{tikz}
\usepackage{tabularx}
\usepackage{array}
\usepackage{eqparbox}
\usepackage{booktabs}

\usepackage{amssymb}
\usepackage{makecell}
\usepackage{diagbox}
\usepackage{mathtools}

\usepackage{array}
\usepackage{lineno}

\usepackage{url}
\usepackage{hyperref}
\usepackage{enumitem}

\newcommand{\pc}{\mathbf{P}}

\newcommand{\pp}{\mathbb{P}}

\newcommand{\weq}{\stackrel{\omega}{=}}
\newcommand{\seq}{\stackrel{s}{=}}

\usepackage{enumitem}

\title{Algebraic, Topological, and Mereological Foundations of Existential Granules}
\titlerunning{Existential Granules}
\author{\textsf{A Mani}\thanks{This research is supported by Woman Scientist Grant No. WOS-A/PM-22/2019 of the Department of Science and Technology.}}
\authorrunning{A Mani}
\institute{Machine Intelligence Unit, Indian Statistical Institute, Kolkata\\
203, B. T. Road, Kolkata-700108, India\\
Email: \texttt{$a.mani.cms@gmail.com$} \texttt{$amani.rough@isical.ac.in$}\\
Homepage: \url{https://www.logicamani.in}\\
Orcid: \url{https://orcid.org/0000-0002-0880-1035}}

\begin{document}

\maketitle

\begin{abstract}
In this research, new concepts of existential granules that determine themselves are invented, and are characterized from algebraic, topological, and mereological perspectives. Existential granules are those that determine themselves initially, and interact with their environment subsequently. Examples of the concept, such as those of granular balls, though inadequately defined, algorithmically established, and insufficiently theorized in earlier works by others, are already used in applications of rough sets and soft computing. It is shown that they fit into multiple theoretical frameworks (axiomatic, adaptive, and  others) of granular computing. The characterization is intended for algorithm development, application to classification problems and possible mathematical foundations of generalizations of the approach. Additionally, many open problems are posed and directions provided.  
\end{abstract}


\section{Introduction}

In the philosophy literature,  existentialism is about basic questions of human existence, individuality, and interactions with the environment. In this research, the adjective \emph{existential} is used in relation to self-determination of objects and subsequent transformations in relation to objects of similar type in their environment. Consequently, they are a form of adaptive objects, and not all adaptive objects are existential. 

In the axiomatic approach to granularity \cite{am5586,am501,am240}, granules are typically specified by conditions. Though external generation procedures are not specified in the literature, they can be expected to be mostly compatible with the methodology. In the precision-based approach \cite{ya01,tyl,gll2006}, precision-levels define admissible granules. However, a mere specification of a precision-level, rarely defines a granule or generates one. Even if it does, it is not that the precision level is intrinsic to the granule. Generation procedures can certainly be added to precision-based granules subject to compatibility with the conditions. Adaptive granules \cite{skajsd2016} are expected to adapt to processes, and even they may be generated initially through some procedures. The existential aspect in these are their initial self-determination, and subsequent transformation in response to interactions. These ideas can be defined in other categorical perspectives of the same theories.     

Closed and open balls and spherical surfaces are pretty standard objects in studies on metrizable spaces and related topologies. However, they are not interpreted as such in some empirical machine learning practices, and if so, how are they interpreted? This is one of the problems tackled in this research.

In recent research, some improved algorithms for K-means clustering and classification issues are developed \cite{ballk20,xldw2019,xdwgg2022}. The algorithms are explained\\ through ideas of \emph{granular balls}. However, they are not sufficiently theorized in the mentioned papers, and in others that make use of the concept \cite{cwyml2021,pwxwc2022,qxhq2023,jpzy2023}. The basic assumptions, and possible generalizations of the concept are investigated by the present author here. New concepts of existential granules are proposed as a severe generalization of the concept, and are shown to be compatible with her axiomatic frameworks for granules \cite{am5586,am501,am240}. A somewhat under-specified example of the formation of existential granules in real life contexts is the following:  When law enforcers try to comb a forested area for possible illegal activity with information from drones and other sources, then they typically form multiple teams to cover distinct sub-areas. Based on the result from the combing operation completed, they are likely to redefine the sub-areas to be searched again. The sequence of set of subareas at each stage will stabilize when all relevant areas are checked. The sub-areas may be regarded as granules that transform themselves at each stage of the operation due to new information, and the state of the search operation(s) on other areas. Finally, they become stable at some later stage.

Clean rough randomness as a not-necessarily stochastic or algorithmically randomness concept is recently introduced by the present author \cite{am23b,am23f}, and is capable of modeling many algorithms such as those that operate over entire sets of tolerances. The problem of precisely formalizing the adaptive aspect of the algorithms using related functions is additionally posed.     

The following sections are organized as follows. Some background is provided in the next section. Existential ball K-Means (BKM) algorithms are analyzed, related partial algebras are invented and a soft generalization where the algorithm works is specified in the third section. The granular ball methodology is formalized through a reading of related algorithms in the next. A formal approach to existential granules is invented in the fifth section and the problem of \emph{appropriately} formalizing the BKM algorithm in the rough randomness perspective is formulated. Further directions are considered in the sixth.

\section{Background}
Distances are often intended to model qualitative ideas of being different in numeric terms. Therefore, it is not required to satisfy most conditions typical of a metric in a metric space.
A \emph{distance function} on a set $S$ is a function $\sigma : S^2 \longmapsto \Re_+$ that satisfies 
\begin{equation}
 (\forall a) \sigma(a,a) = 0 \tag{distance}
\end{equation}
The collection $\mathcal{B} = \{B_\sigma (x, r): x\in S \& r>0 \}$ of all $r$-spheres generated by $\sigma$ is a weak base for the topology $\tau_\sigma$ defined by 
\[V\in \tau_\sigma \text{ if and only if } (\forall x \in V \exists r >0) B_\sigma (x, r)\subseteq V\]

Consider the conditions:
\begin{description}
 \item[Identity]{$(\forall a, b)(\sigma (a, b) = 0 \leftrightarrow a=b)$,}
 \item[Symmetry]{$(\forall a, b) \sigma (a, b) = \sigma (b, a)$,}
 \item[Triangle]{$(\forall a, b, c)\, \sigma (a, b) \leq \sigma (a, c)+\sigma (c, b)$, and}
 \item[Pseudo-Identity]{$(\forall a, b)(\sigma (a, b) = 0 \longrightarrow a=b)$.}
\end{description}
$\sigma$ is said to be a metric or semimetric or pseudometric respectively as it satisfies all the four or identity and symmetry or the last three conditions respectively. A quasi-metric is a distance that satisfies the triangle inequality, while a distance that satisfies the triangle inequality up to a constant $k> 0$ (k-triangle: $(\forall a, b, c)\, k\sigma (a, b) \leq \sigma (a, c)+\sigma (c, b)$) is called a weak quasi-metric. Given a distance function on a set, a topology does not automatically follow. This holds as well for semimetrics \cite{mgsn1973}. All generalized metric spaces will be collectively referred to as $*$-metric spaces.

If $x$ is a point in a $*$-metric space $(X, \sigma)$, and $H$ a subset of $X$ then the distance of $x$ from $H$ is given by $\underline{\sigma}(x, H) = \inf \{\sigma(x, a): a\in H\}$. The distance between two subsets $H$ and $F$ can be measured with the Hausdorff distance $\sigma_h$ or the infimal distance $\sigma_I$(these are not metrics): \[\sigma_h(H, F) = \max \{\sup_{x\in H} \underline{\sigma} (x, F), \sup_{x\in F} \underline{\sigma} (H, x) \} \,\&\, \sigma_I(H, F) = \inf \{\sigma(a, b): a\in H, b\in F\}.\] The former is a metric on the set of compact subsets if $\sigma$ is a metric.

\subsection{Topological Vector Spaces}

Some familiarity with topological vector spaces (TVS)\cite{aw1978,bosm2017} will be assumed. 
A $*$-metric vector space is a pair $(X, \sigma)$, with $X$ being a vector space over the real field, and $\sigma$ a $*$-metric such that the operations are jointly continuous (that is if $(x_n) \rightarrow x$ and $(b_n) \rightarrow b$ in $X$, and $(\alpha_n) \rightarrow \alpha$ in $\Re$, then 
$(\alpha_n x_n + b_n) \rightarrow \alpha x + b$.)

Consider the following properties of a function $p: X \longmapsto \Re_+$ for any $x, b\in X$ 
\begin{align*}
p(0) = 0 \tag{PN1}\\
p(x) \geq 0  \tag{PN2}\\
p(-x) = p(x) \tag{PN3}\\
p(x+b) \leq p(x) + p(b) \tag{PN4}\\
\text{Continuity of scalar multiplication}\tag{PN5}\\
\text{If } p(x) = 0 \text{ then } x= 0 \tag{PNT}\\
p(\alpha x) = |\alpha| p(x) \tag{SN1}
\end{align*}
$p$ is said to be a \emph{paranorm} (respectively \emph{seminorm}) if it satisfies PN1-PN5 (respectively PN2,PN4 and SN1). It is \emph{total}, if it satisfies PNT. All seminorms are paranorms, and a total seminorm is a \emph{norm}.

A paranormed space is a pair $(X, p)$ where $p$ is a paranorm over the vector space $X$. It is complete if $(X,\sigma)$ is complete, where $\sigma (a, b) = p(a - b)$. Every pseudometric vector space can be endowed with a paranorm from which it is derived.

It should be noted that $*$-metrics that have nothing to do with any intended topology are sometimes used in ML practice.

\subsection{Partial Algebraic Systems}

For basics of partial algebras, the reader is referred to \cite{bu,lj}.
\begin{definition}
A \emph{partial algebra} $P$ is a tuple of the form \[\left\langle\underline{P},\,f_{1},\,f_{2},\,\ldots ,\, f_{n}, (r_{1},\,\ldots ,\,r_{n} )\right\rangle\] with $\underline{P}$ being a set, $f_{i}$'s being partial function symbols of arity $r_{i}$. The interpretation of $f_{i}$ on the set $\underline{P}$ should be denoted by $f_{i}^{\underline{P}}$, but the superscript will be dropped in this paper as the application contexts are simple enough. If predicate symbols enter into the signature, then $P$ is termed a \emph{partial algebraic system}.   
\end{definition}

In this paragraph the terms are not interpreted. For two terms $s,\,t$, $s\,\stackrel{\omega}{=}\,t$ shall mean, if both sides are defined then the two terms are equal (the quantification is implicit). $\stackrel{\omega}{=}$ is the same as the existence equality (also written as $\stackrel{e}{=}$) in the present paper. $s\,\stackrel{\omega ^*}{=}\,t$ shall mean if either side is defined, then the other is and the two sides are equal (the quantification is implicit). Note that the latter equality can be defined in terms of the former as 
\[(s\,\stackrel{\omega}{=}\,s \, \longrightarrow \, s\,\stackrel{\omega}{=} t)\&\,(t\,\stackrel{\omega}{=}\,t \, \longrightarrow \, s\,\stackrel{\omega}{=} t) \]

Various kinds of morphisms can be defined between two partial algebras or partial algebraic systems of the same or even different types. For two partial algebras of the same type \[X\, =\, \left\langle\underline{X},\,f_{1},\,f_{2},\,\ldots ,\, f_{n} \right\rangle \text{ and } W\, =\, \left\langle\underline{W},\,g_{1},\,g_{2},\,\ldots ,\, g_{n} \right\rangle ,\] a map $\varphi \, :\, X\, \longmapsto\, W$ is said to be a 
\begin{itemize}
\item {\emph{morphism} if for each $i$, \[(\forall (x_1,\, \ldots \, x_k)\,\in \, dom (f_i)) \varphi (f_{i}(x_1 , \ldots , \, x_k))\,=\,  g_i (\varphi(x_1),\, \ldots , \, \varphi (x_k)) \]}
\item {\emph{closed morphism}, if it is a morphism and the existence of\\ $g_{i} (\varphi(x_1),\, \ldots , \, \varphi (x_k))$ implies the existence of $f_{i}(x_1 , \ldots , \, x_k)$.}
\end{itemize}

Usually it is more convenient to work with closed morphisms.

\section{Existential Granular K-Means Algorithm}

The end product of a hard or soft clustering can often be interpreted as a granulation. The so-called ball granular computing \cite{xldw2019} is not properly formalized from a mathematical perspective as the goal of the authors is to stress the performance of their algorithms. Its origin is obviously related to the ball K-means algorithm \cite{ballk20}. A critical analysis with some generalization of the last mentioned method is proposed first after reconsidering the basic assumptions implicit in it.

Let the dataset of points be $V$ that is a subset of the real topological vector space $X$  with pseudometric (or metric) $\sigma$ which in turn is equivalent to a paranorm. $V$ is not usually closed under the algebraic operations induced from $X$. Algebraic closure on real data is often more complex as additional layers of meaning based on bounds or types may be of interest. Some questions that can shape the semantic domain and therefore relevant algebraic models are
\begin{enumerate}
 \item {Should the value of operations beyond $V$ be considered?}
 \item {Are the interpretations of operations over $X$ meaningful for the context of $V$? To what extent should they be permitted? Does the smallest subspace $Alg(V)$ containing $V$ suffice?}
 \item {Should the interpretations of the operations over $X$ be reinterpreted (at least partly) over $V$. This is especially useful when the values or bounds imposed by $V$ are alone meaningful.}
\end{enumerate}
Depending on the answers to these, the appropriate algebraic operations on the balls may be selected and this lead to a natural generalization of algorithm.

The basic steps of the ball k-means algorithm are 
\begin{description}
\item[Subregion]{Form an arbitrary clustering $E_1, E_2, \ldots E_k $ of $V$.}
\item[MCT]{Compute the mean $c_i$ for each subset $E_i$.}
\item[Radius]{Taking the greatest distance among points in $E_i$ from $c_i$ as the radius $r_i$, generate the ball $C_i$ for each $i$. }
\item[Neighbors]{Define the relation $\eta C_j C_i$ (for $C_j$ is a neighbor of $C_i$) if and only if  $\sigma c_i c_j < 2 r_i$. For $C_i$, let $N_{C_i}$ be its set of neighbor balls (granules).}
\item[Stable]{ If $N_{C_i} \neq \emptyset$, then its stable region is defined by
\begin{itemize}
\item {$St_{\sigma}(C_i) = B(c_i, 0.5\min {\sigma(c_i, c): c\in N_{C_i}})$,}
\item {and its active area by $AA(C_i) = C_i \setminus St(C_i)$.}
\end{itemize}}
\item[Annular Regions]{Let $Card(N_C) = k$, then for $i\in[1, k]$, the ith annular region on $C$, $\mathbb{A}_i^C$ is  $\{x: \sigma(c, c_i) < 2 \sigma (x, c) \leq \sigma(c, c_{i+1})\}$ for $i< k$ and $x: \sigma(c, c_i) < 2 \sigma (x, c) \leq r $ for $i=k$.}
\end{description}

The BKM algorithm and related considerations can be extended to any metric TVS.  However, the results cannot be guaranteed for semi-metric spaces or pseudo-metric spaces, in general.

For each $i$, the radius at the first iteration $r_i = Sup\{ \sigma(x, c_i):\, x\in E_i\}$. Given two ball clusters $C_i$ and $C_j$ with centers $c_i$ and $c_j$ respectively at a fixed iteration level, define \[\eta C_j C_i \text{ iff } \sigma (c_i, c_j) < 2 r_i .\]
$\eta$ is a reflexive, non-symmetric relation in general. $C_j$ is a  \emph{neighbor} of $C_i$ if and only if $\eta C_i C_j$
For $C_i$, let $N_{C_i}$ be its set of neighbor balls (granules). If $N_{C_i} \neq \emptyset$, then its stable region is defined by
\[St(C_i) = B(c_i, 0.5\min {\sigma(c_i, c): c\in N_{C_i}}).\] The active area $AA(C_i) = C_i \setminus St(C_i)$.

The term \emph{$i$-closest} is not defined in the paper \cite{ballk20}. It is simply the closest neighbor cluster(s). As it is based on distance between centers, uniqueness cannot be guaranteed. Let $Card(N_C) = k$, then for $i\in[1, k]$, the ith annular region on $C$, $\mathbb{A}_i^C$ is  $\{x: \sigma(c, c_i) < 2 \sigma (x, c) \leq \sigma(c, c_{i+1})\}$ for $i< k$ and $x: \sigma(c, c_i) < 2 \sigma (x, c) \leq r $ for $i=k$.

The ball $k$-means algorithm can be reinterpreted as a granular approximation procedure of an unknown clustering that is supposed to exist. The steps in the approximation being guided by $\eta$, stable regions, and annular regions as proved in Theorem \ref{mix} (\cite{ballk20}). Stable regions may be read as partial lower approximations of the initial granules at that stage. 

\begin{theorem}\label{mix}
\begin{enumerate}
\item {If $C_i$ is a neighbor of $C$, then some non-stable points of $C$ may be moved into $C_i$.}
\item {If $C_j$ is not a neighbor of $C$, then no points of $C$ can be moved into $C_j$}
\item {For a given $C$ with center $c$, and $Card(N_C)= k$, the points in the $i$th ($i\leq k$)  annular space of $C$ can only be moved within the first $i$-closest neighbor clusters and itself.}
\item {If $c_i^{(t)}$ is the center of the ball $C_i$ in the $t$th iteration, then if $\sigma(c_i^{(t-1)}, c_j^{(t-1)}) \geq 2r_i^{(t)} +\sigma (c_i^{(t)},c_i^{(t-1)}) +\sigma (c_j^{(t)},c_j^{(t-1)})$, then $C_j$ cannot be a neighbor ball of $C_i$ in the current iteration. So the computation of the center distance is avoided.}
\end{enumerate}
\end{theorem}

It can additionally be proved that

\begin{theorem}
If $X$ is a paranormed TVS, then the BKM algorithm terminates in a soft clustering. 
\end{theorem}

\subsection{Partial Algebras for BKM Variants}\label{bkmv} 

Let $B_r^V(c)$ denote the closed ball $\{x: \,x\in V, \&\, \sigma(x, c) \leq r \}$ with center $c$ and radius $r$ over $V$ (it will also be referred to as the \emph{cautious closed ball}), then the following algebraic partial/total operations are definable for any $a, b\in B_r^V(c)$ and any $\alpha, \beta \in \Re$

\[\alpha a\ovee \beta b= 
\begin{dcases}
    \alpha a + \beta b,& \text{if } \alpha a, \beta b,  \alpha a+ \beta b \in B_r^V(c)\\
    \text{undefined},              & \text{otherwise}
\end{dcases}
\]

On the closed ball $B_r^X(c) = \{x: \,x\in X, \&\, \sigma(x, c) \leq r \}$ too, a similar operation $\oplus$ may be interpreted (relative to the operations on $X$).

\[\alpha a\oplus \beta b= 
\begin{dcases}
    \alpha a\oplus \beta b,& \text{if } \alpha a+ \beta b \in B_r^X(c)\\
    \text{undefined},              & \text{otherwise}
\end{dcases}
\]

\begin{theorem}\label{ballalg}
In the above context, $B_r^X(c)$  satisfies 
\begin{align}
a\oplus b \seq b \oplus a   \tag{weak* comm}\\
a\oplus (b \oplus c) \weq (a\oplus b) \oplus c   \tag{weak assoc}\\
\alpha (\beta a) \weq (\alpha \beta)a   \tag{weak scal1}\\
\alpha a\oplus \beta a \seq (\alpha + \beta)a   \tag{weak* scal2}\\
a \oplus 0 \seq 0 \oplus a \tag{weak* 0}\\
(\forall a,b, c)(a\oplus b = 0 = a\oplus c \longrightarrow b=c) \tag{inverse}
\end{align}
\end{theorem}

\begin{proof}
The weak versions of the equalities hold when both sides are defined, while the stronger version ($\seq$) require any one of the sides to be defined.
Weak* commutativity is obvious. Weak associativity holds, and its stronger version does not because $a \oplus b $ may not be defined, even though $a\oplus (b \oplus c)$ is. 
\end{proof}

\begin{theorem}\label{cballalg}
In the above context, $B_r^V(c)$ satisfies $dom(\ovee) \subset dom(\oplus)$ and
\begin{align}
a\ovee b \seq b \ovee a   \tag{weak* comm}\\
a\ovee (b \ovee c) \weq (a\ovee b) \ovee c   \tag{weak assoc}\\
\alpha (\beta a) \weq (\alpha \beta)a   \tag{weak scal1}\\
\alpha a\ovee \beta a \seq (\alpha + \beta)a   \tag{weak* scal2}\\
a \ovee 0 \seq 0 \ovee a \tag{weak* 0}\\
(\forall a,b, c)(a\ovee b = 0 = a\ovee c \longrightarrow b=c) \tag{inverse}
\end{align}
\end{theorem}

\section{The Granular Ball Methodologies}

A version of the granular ball methods for classification can be found in the preprint \cite{xdwgg2022}. Readers are left wondering whether a norm is even being used (Eqn 2 in page 3), while the exact partitioning of parent balls into child balls (in Definition 1) is impossible. However, the algorithm is relatively clear, and examples for the intent of Definition 1 are provided. Earlier versions have additional mathematical issues (see the discussion at
\url{https://pubpeer.com/publications/4354287243FC39A66DD432BC41046B}).\\ The methods for adaptive granules involves quality checks based on purity of granules (relative to proportion of labels), and heterogeneity of overlap of granules. Otherwise, the essential methods are variations of the one used in the ball K-means algorithm.

\begin{enumerate}
\item {The data set may be partly labelled.}
\item {Regard whole data set or a sample $V$ as a closed sphere  with center $c = \frac{1}{n}\sum (x_i)$ and radius as $r = \frac{1}{n} \sum \sigma(x_i, c)$.} 
\item {Split the current granular ball into k sub balls using ball K-means. It should be noted that splitting is not a partitioning operation.}
\item {Check quality of granular Balls through simple purity measures based on ratio of majority label.}
\item {Stop if the purity measure is OK. Otherwise, repeat the process on the balls derived. } 
\end{enumerate}

The adaptive version is similar but with further stages of splitting whenever a granular ball at the current iteration has heterogeneous overlap with another granular ball, and an accelerated granular ball generation process. Child balls and parent balls are further used in the quality checks, and the ball K-means algorithm is avoided.

\subsection{Fixing the Mathematics}

From the intent of definition 1 (of parent and child balls), it is clear to the present author that the \emph{real data points within a ball} are confused with the ball. A mathematical way of correcting can be through differential geometry, and at least concepts of orientation are essential. The code and algorithm are however based on distances from centers, and labels. A minimal fix that avoids the geometry is the following:

\begin{definition}
Let $V$ be a finite subset of a real normed finite dimensional TVS $X$, and $B_r^V(c)$ be a ball, and $\{B_{r_i}^V(c_i)\}$ a finite sequence of $n$ number of balls, all interpreted in $V$ (with centers in $X$) then $B_r^V(c)$ is a \emph{major ball} and $\{B_{r_i}^V(c_i)\}$ are \emph{minor balls} if and only if the following holds:
\begin{align*}
\bigcup_i B_{r_i}^V(c_i) = \{x: x\in B_{r_i}^V(c_i) \text{ for any }i \} = B_r^V(c) \tag{sum}\\
\bigcap_i \, B_{r_i}^V(c_i)  = \emptyset  \tag{collectionwise disjoincy}
\end{align*}
\end{definition}

This definition makes no sense when $B_r^V(c) = B_r^X(c)$ (and actually under much weaker conditions). Further, minor/child balls should rather be pairwise disjoint (for any $i\neq j \, B_{r_i}^V(c_i) \cap B_{r_j}^V(c_j) = \emptyset$) in Definition 1 of \cite{xdwgg2022}.

\section{Existential Granulations}

To accommodate multiple nonequivalent concepts of granules, and granulations, a loose definition of \emph{existential granule} is proposed first. 
Suppose $X = \left\langle \underline{X}, \mathcal{F}, \mathcal{G} \right\rangle$ is a triple with $\underline{X}$ being a set, $\mathcal{F}$ a mathematical structure on it, and $\mathcal{G} \subseteq \wp (\underline{X})$ a granulation on it. A granule $G\in \mathcal{G} $ will be said to be \emph{existential} if and only if there exists a subset $E$ of $G$ and an operator $\Game: \wp(X) \longmapsto \wp(X)$ such that $G= \Game (E)$ and $(\exists n\geq 1) \Game ^{n+1}(E) = \Game^{n}(E)$.
A granulation is existential, if it is a collection of existential granules determined by the features encoded by its constituent points. That is, it is essentially self-determining up to a point.  It is possible to argue on this concept being existential in numerous ways -- it is exactly the reason for naming it \emph{existential} as opposed to \emph{self-determined}. The idea of non-crisp granules is known in both the axiomatic, precision-based and adaptive theories of granularity. Existential granules have a precise generation aspect motivated by the problem of reducing computational load. Formalization in the axiomatic abstract perspective requires an additional closure operator as defined below, while computational aspects require further specialization. The definitions below are relatively more convenient for abstract approaches \cite{am5586,am501,am240}. \emph{The main questions of this approach are about formalizing the known applications, representing the operation $\Game$, and suitability of the restrictions of admissible granules}. 

\begin{definition}\label{partapp}
An \emph{high mereological approximation Space} (\textsf{mash}) $\mathbb{S}$ is a partial algebraic system  of the form $\mathbb{S} \, =\, \left\langle \underline{\mathbb{S}}, l , u, \pc, \leq , \vee,  \wedge, \bot, \top \right\rangle$ with $\underline{\mathbb{S}}$ being a set, $l, u$ being operators $:\underline{\mathbb{S}}\longmapsto \underline{\mathbb{S}}$ satisfying the following ($\underline{\mathbb{S}}$ is replaced with $\mathbb{S}$ if clear from the context. $\vee$ and $\wedge$ are idempotent partial operations and $\pc$ is a binary predicate.):
\begin{align*}
(\forall x) \pc xx \tag{PT1}\\
(\forall x, b) (\pc xb \, \&\, \pc bx \longrightarrow x = b) \tag{PT2}\\
(\forall a, b) a\vee b \stackrel{\omega}{=} b\vee a  ; \;  (\forall a, b) a\wedge b \stackrel{\omega}{=} b\wedge a \tag{G1}\\
(\forall a, b) (a\vee b) \wedge a \stackrel{\omega}{=} a  ; \;  (\forall a, b) (a\wedge b) \vee a \stackrel{\omega}{=} a \tag{G2}\\
(\forall a, b, c) (a\wedge b) \vee c \stackrel{\omega}{=} (a\vee c) \wedge (b\vee c) \tag{G3}\\
(\forall a, b, c) (a\vee b) \wedge c \stackrel{\omega}{=} (a\wedge c) \vee  (b\wedge c) \tag{G4}\\
(\forall a, b) (a\leq b \leftrightarrow a\vee b = b \,\leftrightarrow\, a\wedge b = a  ) \tag{G5}\\
(\forall a \in \mathbb{S})\,  \pc a^l  a\,\&\,a^{ll}\, =\,a^l \,\&\, \pc a^{u}  a^{uu}  \tag{UL1}\\
(\forall a, b \in \mathbb{S}) (\pc a b \longrightarrow \pc a^l b^l \,\&\,\pc a^u  b^u) \tag{UL2}\\
\bot^l\, =\, \bot \,\&\, \bot^u\, =\, \bot \,\&\, \pc \top^{l} \top \,\&\,  \pc \top^{u} \top  \tag{UL3}\\
(\forall a \in \mathbb{S})\, \pc \bot a \,\&\, \pc a \top    \tag{TB}
\end{align*}
\end{definition}

In a \emph{high general granular operator space} (\textsf{GGS}), defined below, aggregation and co-aggregation operations ($\vee, \,\wedge$) are conceptually separated from the binary parthood ($\pc$), and a basic partial order relation ($\leq$). Parthood is assumed to be reflexive, antisymmetric, and not necessarily transitive. It may satisfy additional generalized transitivity conditions in many contexts. Real-life information processing often involves many non-evaluated instances of aggregations (fusions), commonalities (conjunctions) and implications because of laziness or supporting  metadata or for other reasons  -- this justifies the use of partial operations. Specific versions of a \textsf{GGS} and granular operator spaces have been studied in the research paper \cite{am501}. Partial operations in \textsf{GGS} permit easier handling of adaptive granules \cite{skajsd2016} through morphisms-- concrete methods need to use the frameworks of clear rough random functions. Note further that it is not assumed that $\pc a^{uu} a^u$. The universe $\underline{\mathbb{S}}$ may be a set of collections of attributes, labeled or unlabeled objects among other things. A \emph{high general existential granular operator space} (\textsf{eGGS}) can be obtained from a GGS by simply restricting the predicate $\gamma$ as follows:
\[\gamma x \text{ if and only if } x\in \mathcal{G} = \Im(\Game)\]
\begin{definition}\label{gfsg}
A \emph{High General Granular Operator Space} (\textsf{GGS}) $\mathbb{S}$ is a partial algebraic system  of the form \[\mathbb{S} \, =\, \left\langle \underline{\mathbb{S}}, \gamma, l , u, \pc, \leq , \vee,  \wedge, \bot, \top \right\rangle\] with $\mathbb{S} \, =\, \left\langle \underline{\mathbb{S}}, l , u, \pc, \leq , \vee,  \wedge, \bot, \top \right\rangle$ being a \textsf{mash}, $\gamma$ being a unary predicate that determines $\mathcal{G}$ (by the condition $\gamma x$ if and only if $x\in \mathcal{G}$) 
an \emph{admissible granulation}(defined below) for $\mathbb{S}$. Further, $\gamma x$ will be replaced by $x \in \mathcal{G}$ for convenience.
Let $\pp$ stand for proper parthood, defined via $\pp ab$ if and only if $\pc ab \,\&\,\neg \pc ba$). A granulation is said to be admissible if there exists a term operation $t$ formed from the weak lattice operations such that the following three conditions hold:
\begin{align*}
(\forall x \exists
x_{1},\ldots x_{r}\in \mathcal{G})\, t(x_{1},\,x_{2}, \ldots \,x_{r})=x^{l} \\
\tag{Weak RA, WRA} \mathrm{and}\: (\forall x\, \exists
x_{1},\,\ldots\,x_{r}\in \mathcal{G})\,t(x_{1},\,x_{2}, \ldots \,x_{r}) =
x^{u},\\
\tag{Lower Stability, LS}{(\forall a \in
\mathcal{G})(\forall {x\in \underline{\mathbb{S}}) })\, ( \pc ax\,\longrightarrow\, \pc ax^{l}),}\\
\tag{Full Underlap, FU}{(\forall
x,\,a \in\mathcal{G} \exists
z\in \underline{\mathbb{S}}) \, \pp xz,\,\&\,\pp az\,\&\,z^{l}\, =\,z^{u}\, =\,z,}
\end{align*}
\end{definition}
\begin{definition}\label{egfsg}
A \emph{High General Existential Granular Operator Space} (\textsf{eGGS}) $\mathbb{S}$ is a GGS in which the predicate $\gamma$ is replaced by a unary operation $\Game$ that satisfies  
\begin{align*}
\gamma x \text{ if and only if } x\in \mathcal{G} = \Im(\Game) \tag{G1}\\
(\forall x)(\exists n \geq 1) \Game^{n+1}(x) = \Game^{n} (x) \tag{G2}  
\end{align*}
\end{definition}
Existential granular versions of the following particular classes can be defined by analogy. 
\begin{definition}
\begin{itemize}
\item {In the above definition, if the anti-symmetry condition \textsf{PT2} is\\ dropped, then the resulting system will be referred to as a \emph{Pre-GGS}. If the restriction $\pc a^l  a$  is removed from \textsf{UL1} of a \textsf{pre-GGS}, then it will be referred to as a \emph{Pre*-GGS}.}
\item {In a \textsf{GGS} (resp \textsf{Pre*-GGS}), if the parthood is defined by $\pc ab$ if and only if $a \leq b$ then the \textsf{GGS} is said to be a \emph{high granular operator space} \textsf{GS} (resp. \textsf{Pre*-GS)}.}
\item {A \emph{higher granular operator space} (\textsf{HGOS}) (resp \textsf{Pre*-HGOS}) $\mathbb{S}$ is a \textsf{GS} (resp \textsf{Pre*-GS}) in which the lattice operations are total.}
\item {In a higher granular operator space, if the lattice operations are set theoretic union and intersection, then the \textsf{HGOS} (resp. Pre*-HGOS) will be said to be a \emph{set HGOS} (resp. \emph{set Pre*-HGOS}). In this case, $\underline{\mathbb{S}}$ is a subset of a power set, and the partial algebraic system reduces to $\mathbb{S} \, =\, \left\langle \underline{\mathbb{S}}, \gamma, l , u, \subseteq , \cup,  \cap, \bot, \top \right\rangle$ with $\underline{\mathbb{S}}$ being a set, $\gamma$ being a unary predicate that determines $\mathcal{G}$ (by the condition $\gamma x$ if and only if $x\in \mathcal{G}$). Closure under complementation is not guaranteed in it. }
\end{itemize}
\end{definition}

\subsection{Clean Rough Randomness and Models of Algorithms}

Some essential aspects of clean rough randomness \cite{am23b,am23f} are repeated for convenience, and the problem of formalizing the studied algorithms is in the perspective is formulated. 

Many types of randomness are known in the literature. Stochastic randomness, often referred to as randomness, is often misused without proper justification. In the paper \cite{ank1986}, a phenomenon is defined to be \emph{stochastically random} if it has probabilistic regularity in the absence of other types of regularity. In this definition, the concept of regularity may be understood as \emph{mathematical regularity} in some sense. Generalizations of mathematical probability theory through hybridization with rough sets from a stochastic perspective are explained in the book \cite{bliu2004}. This approach is not ontologically consistent with pure rough reasoning or explainable AI as its focus is on modeling the result of numeric simplifications in a measure-theoretic context. 

Empirical studies show that humans cannot estimate measures of stochastic randomness and weakenings thereof in real life properly \cite{lbg1994}. This is consistent with the observation that connections in the rough set literature between specific versions of rough sets and subjective probability theories (Bayesian or frequentist) are not good approximations. In fact,  rough inferences are grounded in some non-stochastic comprehension of attributes (their relation with the approximated object in terms of number or relative quantity and quality) \cite{amedit,ppm2}.

The idea of \emph{rough randomness} is defined by the present author \cite{am23b} as follows:\emph{a phenomenon is clean roughly random (C-roughly random) if it can be modeled by general rough sets or a derived process thereof}. In concrete situations, such a concept should be realizable in terms of C-roughly random functions or predicates defined below (readers should note that any one of the concepts of rough objects in the literature \cite{am501} such as \emph{a non crisp object} or a \emph{pair of definite objects of the form $(a, b)$ satisfying} $\pc ab$ among others are permitted): 

\begin{definition}\label{roran1}
Let $\mathcal{A}_\tau$ be a collection of approximations of type $\tau$, and $E$ a collection of rough objects defined on the same universe $S$, then by a \emph{C-rough random function of type-1} (\textsf{CRRF1}) will be meant a partial function \[\xi : \mathcal{A}_\tau \longmapsto E .\] 
\end{definition}

\begin{definition}\label{roran2}
Let $\mathcal{A}_\tau$ be a collection of approximations of type $\tau$, $\mathcal{S}$ a subset of $\wp(S)$, and $\Re$ the set of reals, then by a \emph{C-rough random function of type-2} (\textsf{CRRF2}) will be meant a function \[\chi : \mathcal{A}_\tau \times \mathcal{S}\longmapsto \Re .\] 
\end{definition}

\begin{definition}\label{roran3}
Let $\mathcal{A}_\tau$ be a collection of approximations of type $\tau$, and $F$ a collection of objects defined on the same universe $S$, then by a \emph{C-rough random function of type-3} (\textsf{CRRF3}) will be meant a function \[\mu : \mathcal{A}_\tau \longmapsto F .\] 
\end{definition}

\begin{definition}\label{hroran}
Let $\mathcal{O}_\tau$ be a collection of approximation operators of type $\tau_l$ or $\tau_u$, and $E$ a collection of rough objects defined on the same universe $S$, then by a \emph{C-rough random function of type-H} (\textsf{CRRFH}) will be meant a partial function \[\xi : \mathcal{O}_\tau \times \wp(S) \longmapsto E .\] 
\end{definition}

It is obvious that a CRRF1 and CRRF2 are independent concepts, while a total CRRF1 is an CRRF3, and CRRFH is distinct (though related to CRRF3). The set of all such functions will respectively be denoted by $CRRF1(S, E, \tau)$,  $CRRF2(S, \Re , \tau)$, $CRRF3(S, F, \tau)$, and $CRRFH(S, E, \tau)$. For detailed examples, the reader is referred to the earlier papers \cite{am23b,am23f}

\begin{example}\label{ex1}
Let $S$ be a set with a pair of lower ($l$) and upper ($u$)  approximations satisfying (for any $a, b, x \subseteq S$)
\begin{align}
x^l \subseteq x^u \tag{int-cl}\\
x^{ll} \subseteq x^l\tag{l-id}\\
a\subseteq b \longrightarrow a^l \subseteq b^l \tag{l-mo}\\
a\subseteq b \longrightarrow a^u \subseteq b^u \tag{u-mo}\\
\emptyset^l = \emptyset \tag{l-bot}\\
S^u = S\tag{u-top}
\end{align}
The above axioms are minimalist, and most general approaches satisfy them.

In addition, let 
\begin{align}
\mathcal{A}_\tau = \{x : (\exists a \subseteq S)\, x= a^l \text{ or }  x= a^u  \tag{1}\\
E_1 = \{(a^l, a^u):\, a\in S\} \tag{E1}\\
F = \{a: \, a\subseteq S \& \neg \exists b b^l = a \vee b^u = a\} \tag{E0}\\
E_2 = \{b: b^u = b \& b\subseteq S\}\tag{E2}\\
\xi_1 (a) = (a, b^u) \text{ for some } b\subseteq S\tag{xi1}\\
\xi_2 (a) = (b^l, a ) \text{ for some } b\subseteq S \tag{xi2}\\
\xi_3 (a) = (e, f) \in E_1 \,\&\, e= a \text{ or } f= a \tag{xi3}
\end{align}

$E_1$ in the above is a set of rough objects, and a number of algebraic models are associated with it \cite{am501}. A partial function $f: \mathcal{A}_\tau\longmapsto E_1$ that associates $a\in \mathcal{A}_\tau$ with a minimal element of $E_1$ that covers it in the inclusion order is a CRRF of type 1. For general rough sets, this CRRF can be used to define algebraic models and explore duality issues \cite{am5019}, and for many cases associated these are not investigated. A number of similar maps with value in understanding models \cite{amedit} can be defined. Rough objects are defined and interpreted in a number of other ways including $F$ or $E_2$.

Conditions \textsf{xi1-xi3} may additionally involve constraints on $b$, $e$ and $f$. For example, it can be required that there is no other lower or upper approximation included between the pair or that the second component is a minimal approximation covering the first. It is easy to see that 

\begin{theorem}
$\xi_i$ for $i=1, 2, 3$ are CRRF of type-1. 
\end{theorem}
\end{example}
\begin{example}
In the above example, rough inclusion functions, membership, and quality of approximation functions \cite{js09,ag2009} can be used to define CRRF2s. An example is the function $\xi_5$ defined by 
\begin{equation}
 \xi_5(a, b) = \dfrac{Card(b\setminus a)}{Card(b)}
\end{equation}
\end{example}

\subsection{Formalizing the BKM Algorithms}

The ball K-means algorithm can potentially be formalized by rough random functions of type 3 in several ways. For this purpose, one can use a single RRF-3 $\varphi$ and a number of classical lower and upper approximation that describe each update on the original $k$ clusters sequentially or use a sequence of RRF-3s with pairs of classical lower and upper approximations to describe the updates. Therefore, the real problem is of finding and formulating the most appropriate formalization. How does one restrict the choice of approximation operations?  

All crisp clusterings form partitions, and therefore all such clusterings form the granulation of Pawlak rough sets over the universe in question. This is the suggested origin of the \emph{classical and upper approximations}. 

\section{Further Directions}

It might appear to easy to cast the ball K-means and granular ball algorithms in the interactive granular computing perspective. It is already shown that such is not essential. The proposal of interactive granules and related computing (IGrC) is formulated in relation to a certain perception of the basic semantic domain, and is primarily intended to reduce the complexity of decision-making in application contexts \cite{wav2008,skajsd2016}.  Some objects are supposed to be \emph{non-mathematical objects} at a level of discourse, and possess some properties of granularity. The use of complex granule (c-granule) comprising abstract objects, physical objects, as well as objects linking abstract and physical objects, by itself, and their rule-based approach apparently constrains the authors  to that view. States of c-granules are represented by networks of informational granules (ic-granules) linking abstract and physical objects. Such c-granules are intended for modeling perceptions of physical processes in the real world. However, the mathematical approach to such cases is through improved sequences of models, and objects, and through better choice of semantic domains. Data drives nothing, it is for us to invent models that make any driving to be possible at all.

The obvious idea of replacing the hyper-sphere with a smooth hypersurface is possible in theory, and justifiable if the geometry is relevant. However, the computational complexity may increase substantially. Actually, no hyperspheres or balls are used in the both the BKM and granular ball algorithms. It is only in the imagination of the authors. If the shapes generated by points are really of interest then other metrics and the Hausdorff-Gromov distance \cite{fm2013} may be used painfully. The possible mathematical generalization of the proposed method requires justification in applied problems. The geometrical shape of granules typically matter in the domain of topological data analysis \cite{kimm23}, spatial mereology \cite{ham2017} and near sets \cite{jfp}. Such approaches have steep requirements on the domain for easier computing. 

In future studies, existential granules will be explored in greater depth.

\bibliographystyle{splncs04.bst}
\bibliography{algrough23+}

\begin{thebibliography}{10}
\providecommand{\url}[1]{\texttt{#1}}
\providecommand{\urlprefix}{URL }
\providecommand{\doi}[1]{https://doi.org/#1}

\bibitem{lbg1994}
Beach, L., Braun, G.: Laboratory studies of subjective probability: a status
  report. In: Wright, G., Ayton, P. (eds.) Subjective Probability, pp.
  107--128. John Wiley (1994)

\bibitem{bosm2017}
Bogachev, V.I., Smolyanov, O.G.: Topological Vector Spaces and Their
  Applications. Springer Monographs in Mathematics, Springer (2017)

\bibitem{ham2017}
Burkhardt, H., Seibt, J., Imaguire, G., Gerogiorgakis, S. (eds.): {Handbook of
  Mereology}. Philosophia Verlag, Germany (2017)

\bibitem{bu}
Burmeister, P.: {A Model-Theoretic Oriented Approach to Partial Algebras}.
  Akademie-Verlag (1986, 2002)

\bibitem{cwyml2021}
Chen, Y., Wang, P., Yang, X., Mi, J., Liu, D.: Granular ball guided selector
  for attribute reduction. Knowledge-Based Systems  \textbf{229} (2021).
  \doi{10.1016/j.knosys.2021.107326}

\bibitem{mgsn1973}
Gagrat, M., Naimpally, S.: Proximity approach to semi-metric and developable
  spaces. Pacific Journal of Mathematics  \textbf{44}(1),  93--105 (1973)

\bibitem{ag2009}
Gomolinska, A.: {Rough Approximation Based on Weak q-RIFs}. Transactions on
  Rough Sets  \textbf{X},  117--135 (2009)

\bibitem{jpzy2023}
Ji, X., Peng, J., Zhao, P., Yao, S.: Extended rough sets model based on fuzzy
  granular ball and its attribute reduction. Information Sciences  (2023).
  \doi{10.1016/j.ins.2023.119071}

\bibitem{kimm23}
Kim, W., Memoli, F.: Persistence over posets. Notices of the American
  Mathematical Society  \textbf{2761},  1214--1224 (September 2023).
  \doi{10.1090/noti2761}

\bibitem{ank1986}
Kolmogorov, A.N.: On the logical foundations of probability theory. In:
  Shiryayev, A.N. (ed.) Selected Works of A. N. Kolmogorov, vol.~2, chap.~53,
  pp. 515--519. Kluwer Academic, Nauka (1986)

\bibitem{tyl}
Lin, T.Y.: {Granular Computing-1: The Concept of Granulation and Its Formal
  Model}. Int. J. Granular Computing, Rough Sets and Int Systems
  \textbf{1}(1),  21--42 (2009)

\bibitem{bliu2004}
Liu, B.: Uncertainty Theory, Studies in Fuzziness and Soft Computing, vol.~154.
  Springer (2004)

\bibitem{gll2006}
Liu, G.: {The Axiomatization of The Rough Set Upper Approximation Operations}.
  Fundamenta Informaticae  \textbf{69}(23),  331--342 (2006)

\bibitem{lj}
Ljapin, E.S.: {Partial Algebras and Their Applications}. Academic, Kluwer
  (1996)

\bibitem{am240}
Mani, A.: {Dialectics of Counting and The Mathematics of Vagueness}.
  Transactions on Rough Sets  \textbf{XV}(LNCS 7255),  122--180 (2012)

\bibitem{am501}
Mani, A.: {Algebraic Methods for Granular Rough Sets}. In: Mani, A.,
  D{\"u}ntsch, I., Cattaneo, G. (eds.) {Algebraic Methods in General Rough
  Sets}, pp. 157--336. {Trends in Mathematics}, Birkhauser Basel (2018)

\bibitem{am5019}
Mani, A.: {Representation, Duality and Beyond}. In: Mani, A., D{\"u}ntsch, I.,
  Cattaneo, G. (eds.) {Algebraic Methods in General Rough Sets}, pp. 459--552.
  {Trends in Mathematics}, Birkhauser Basel (2018)

\bibitem{am5586}
Mani, A.: {Comparative Approaches to Granularity in General Rough Sets}. In:
  Bello, R., et~al. (eds.) {IJCRS 2020}, {LNAI}, vol. 12179, pp. 500--518.
  Springer (2020)

\bibitem{am23b}
Mani, A.: Rough randomness and its application. Journal of the Calcutta
  Mathematical Society pp. 1--15 (2023). \doi{10.5281/zenodo.7762335},
  \url{https://zenodo.org/record/7762335}

\bibitem{amedit}
Mani, A., D{\"u}ntsch, I., Cattaneo, G. (eds.): {Algebraic Methods in General
  Rough Sets}. {Trends in Mathematics}, Birkhauser Basel (2018).
  \doi{10.1007/978-3-030-01162-8}

\bibitem{am23f}
Mani, A., Mitra, S.: Large minded reasoners for soft and hard cluster
  validation –some directions. Annals of Computer and Information Sciences,
  PTI pp. 1--16 (2023)

\bibitem{fm2013}
Memoli, F.: The gromov-hausdorff distance: a brief tutorial on some of its
  quantitative aspects. Actes des rencontres du C.I.R.M.  \textbf{3}(1),
  89--96 (2013)

\bibitem{ppm2}
Pagliani, P., Chakraborty, M.: {A Geometry of Approximation: Rough Set Theory:
  Logic, Algebra and Topology of Conceptual Patterns}. Springer, Berlin (2008)

\bibitem{wav2008}
Pedrycz, W., Skowron, A., Kreinovich, V.: Handbook of Granular Computing. John
  Wiley (2008)

\bibitem{pwxwc2022}
Peng, X., Wang, P., Xia, S., Wang, C., Chen, W.: Vpgb: A granular-ball based
  model for attribute reduction and classification with label noise.
  Information Sciences  \textbf{611},  504--521 (2022).
  \doi{10.1016/j.ins.2022.08.066}

\bibitem{jfp}
Peters, J.F.: {Topology of Digital Images-Visual Pattern Discovery in Proximity
  Spaces}. {Intelligent Systems Reference Library, Volume 63}, Springer (2014).
  \doi{10.1007/978-3-642-53845-2}

\bibitem{qxhq2023}
Qian, W., Xu, F., Huang, J., Qian, J.: A novel granular ball computing-based
  fuzzy rough set for feature selection in label distribution learning.
  Knowledge-Based Systems  \textbf{278} (2023).
  \doi{10.1016/j.knosys.2023.110898}

\bibitem{skajsd2016}
Skowron, A., Jankowski, A., Dutta, S.: {Interactive granular computing}.
  Granular Computing  \textbf{1}(2),  95--113 (2016)

\bibitem{js09}
Stepaniuk, J.: {Rough-Granular Computing in Knowledge Discovery and Data
  Mining}. {Studies in Computational Intelligence,Volume 152}, Springer-Verlag
  (2009). \doi{10.1007/978-3-540-70801-8}

\bibitem{aw1978}
Wilansky, A.: Modern Methods in Topological Vector Spaces. McGraw-Hill, New
  York (1978)

\bibitem{xdwgg2022}
Xia, S., Dai, X., Wang, G., Gao, X., Giem, E.: An efficient and adaptive
  granular-ball generation method in classification problem. Arxiv  (2022)

\bibitem{xldw2019}
Xia, S., Liu, Y., Ding, X., Wang, G., Yu, H., Luo, Y.: Granular ball computing
  classifiers for efficient, scalable and robust learning. Information Sciences
   \textbf{483},  136--152 (2019)

\bibitem{ballk20}
Xia, S., Peng, D., Meng, D., Zhang, C., Wang, G., Giem, E., Wei, W., Chen, Z.:
  A fast adaptive k-means algorithm. IEE Transactions on Pattern Analysis and
  Machine Intelligence pp. 1--13 (July 2020). \doi{10.1109/TPAMI.2020.3008694}

\bibitem{ya01}
Yao, Y.Y.: {Information Granulation and Rough Set Approximation}. Int. J. of
  Intelligent Systems  \textbf{16},  87--104 (2001)

\end{thebibliography}

\end{document}